\documentclass[12pt]{article}
\usepackage{graphicx,subfigure}
\usepackage{amsmath}
\usepackage{amsfonts}
\usepackage{amssymb}
\usepackage{url}
\newtheorem{theorem}{Theorem}[section]

\newtheorem{remark}[theorem]{Remark}

\numberwithin{equation}{section}

\newenvironment{proof}[1][Proof]{\textbf{#1.} }{\ \rule{0.5em}{0.5em}}
\usepackage[toc,page]{appendix}

\usepackage{color}

\topmargin=-0.5in \everymath{\displaystyle}
\begin{document}
\baselineskip=18pt

\pagenumbering{arabic}

\begin{center}
{\Large {\bf Spatial Spread of Epidemic Diseases in Geographical Settings: Seasonal Influenza Epidemics in Puerto Rico}}

\vspace{.3in}

Pierre Magal \textsuperscript{1},
G.F. Webb \textsuperscript{2},
Yixiang Wu \textsuperscript{2}

\vspace{.3in}

\begin{footnotesize}
\textbf{1} Mathematics Department, University of Bordeaux, Bordeaux, France
\\
\textbf{2} Mathematics Department, Vanderbilt University, Nashville, TN\\
\end{footnotesize}

\bigskip

\vspace{0.2in}
\end{center}

\begin{abstract}
Deterministic models are developed for the spatial spread of epidemic diseases in geographical settings. The models are focused on outbreaks that arise from a small number of infected hosts imported into subregions of the geographical settings. The goal is to understand how spatial heterogeneity influences the transmission dynamics of the susceptible and infected populations. The models consist of systems of partial differential equations with diffusion terms describing the spatial spread of the underlying microbial infectious agents. The model is compared with real data from seasonal influenza epidemics in Puerto Rico.
\end{abstract}

\medskip


\section{Introduction}
Epidemic outbreaks evolve in geographical regions with considerable variability in spatial locations. This spatial variability is important in understanding the impact of public health policies and interventions in controlling these epidemics. A major difficulty in developing models to describe spatial variability in epidemics is accounting for the movement of people in spatial contexts. Many efforts to develop realistic descriptions of epidemics in geographical settings have used individual based models (IBM). These models employ large-scale societal data of human movement and interaction to simulate human behavior  at spatial and temporal levels based on probabilistic assumptions.  These models require intensive informational input, as well as intensive computational output. Our objective is to provide an alternative approach for modeling spatial epidemics based on deterministic  models formulated as partial differential equations in spatial domains.

Our specific focus is upon seasonal influenza outbreaks in geographical regions.
Seasonal influenza epidemics recur annually during the cold half of the year in each hemisphere. Each annual flu season is normally associated with a major influenza-virus subtype. The associated subtype changes each year, due to development of immunological resistance to a previous year's strain through exposure and vaccinations, and mutational changes in previously dormant viral strains. The beginning activity in each season varies by location, and evolves characteristically in the larger spatial domain. The exact mechanism behind the seasonal nature of influenza outbreaks is unknown.

This paper is organized as follows: In Section 2 we formulate a general deterministic model for the evolution of an epidemic outbreak in a spatial domain. In Section 3 we specify the model to seasonal influenza epidemics in Puerto Rico, and simulate these epidemics in 2015-2016 and 2016-2017. 
In Section 4 we discuss our results and compare them to IBM formulations of spatial epidemics. 
In the Appendix we state and proof theorems for our deterministic model of a spatial epidemic.

\section{A General Deterministic Spatial Epidemic Model}  
Partial differential equations with diffusion terms have been proposed  by many authors to describe the movement of people in various applications, including \cite{fitzgibbon2008simple, ruan2007spatial, rass2003spatial, webb1981reaction, capasso1978global}. In most applications, however, diffusion does not provide a valid description of the way people move in societal settings. Diffusion provides only an averaging process that cannot account for the extreme spatial and temporal heterogeneity in human movement. We argue, alternatively, that the spatial movement of the micro-organisms causing the epidemic, rather than the spatial movement of humans, is an effective way to account for epidemic spatial development. The movement of the infectious agent can be viewed indirectly, as the movement of infectious individuals, described with diffusion processes. It is clear that the contributions of local-distance and long-distance transmission are both involved in the spatial evolution of epidemics. In \cite{gog2014}, however,  it is argued that for the 2009 H1N1 influenza epidemic, local transmission was of greater importance than distant transmissions, as outbreaks in proximate communities resulted in  successful infection chains, whereas, distant transmissions died out after a small number of generations. The underlying assumption is that most infections occur close to home-base of infectious individuals, which spread to nearby susceptible individuals.

Our model has the following formulation:
Suppose that $\Omega\subset \mathbb{R}^2$ is a bounded domain. Let $S(t,{\bf x})$ and $I(t,{\bf x})$ be the spatial densities at location ${\bf x} \in \Omega$ and at  time $t$ of susceptible and infected individuals, respectively. 
\begin{eqnarray}
\label{PRmodel}
\frac{\partial}{\partial t} S(t,{\bf x})  &=&  
  - \frac{ \tau({\bf x}) \,  I(t,{\bf x})^p}{1+ \kappa ({\bf x}) \,  I(t,{\bf x})^q} S(t,{\bf x}) \label{Eq-1}, \, \hspace{5cm} {\bf x} \in \Omega, \, t >0   \\  
\frac{\partial}{\partial t} I(t,{\bf x})  &=& \alpha({\bf x}) \Delta I(t,{\bf x})
  + \frac{ \tau ({\bf x}) \,  I(t,{\bf x})^p}{1+ \kappa ({\bf x}) \,  I(t,{\bf x})^q}  S(t,{\bf x}) - \lambda({\bf x}) I(t,{\bf x}), \, \, {\bf x} \in \Omega, \, t >0   \label{Eq-2} \\
\frac{\partial}{\partial \eta} I(t,{\bf x}) &=& 0, \hspace{9.0cm} {\bf x} \in \partial \Omega, \, t >0   \label{Eq-3} \\
S(0,{\bf x}) &=& S_0({\bf x}),\, \, I(0,{\bf x}) = I_0({\bf x}), \,\,\  \hspace{5.5cm}{\bf x} \in  \Omega  \label{Eq-4}
\end{eqnarray}
where $\alpha({\bf x})$ is the diffusion parameter for infected individuals, $\tau({\bf x}) , \kappa({\bf x}) , p$ and $q$ are  transmission parameters, and $\lambda({\bf x})$ is the removal rate of infected individuals. The transmission rate has nonlinear incidence form (\cite{liu1987}, \cite{hethcote1991}, \cite{ruan2003}), where $\tau I(t,{\bf x})^p$ measures the force of infectiousness and $1 / (1 + \kappa I(t,{\bf x})^q )$  measures reduced infectiousness resulting from behavioral change as the number of infected individuals increases. The parameter  $\alpha({\bf x})$, $\kappa({\bf x}), \tau$({\bf x}), and $\lambda({\bf x})$  are positive continuous functions on $\overline\Omega$, and the initial data $S_0$ and $I_0$ are nonnegative continuous functions on $\overline\Omega$. 


\section{ Seasonal Influenza Epidemics in Puerto Rico}

In this section we simulate  the 2015-2016 and 2016-2017 seasonal influenza epidemics in Puerto Rico. The  island of Puerto Rico consists of 76 municipalities, with total population of almost 3,500,000, in a geographical region of approximately 170 km  by 60 km. The four major municipalities with largest population are the eastern San Juan (population 2,350,000),  the southern Ponce (population 262, 000), the western Arecibo (population 193,000), and the western Mayag\H{u}ez (population 89,000 (see Figure \ref{Fig1a-1b}).

\begin{figure}
\begin{center}
\subfigure
{\includegraphics[width=6.5in,height=2in]{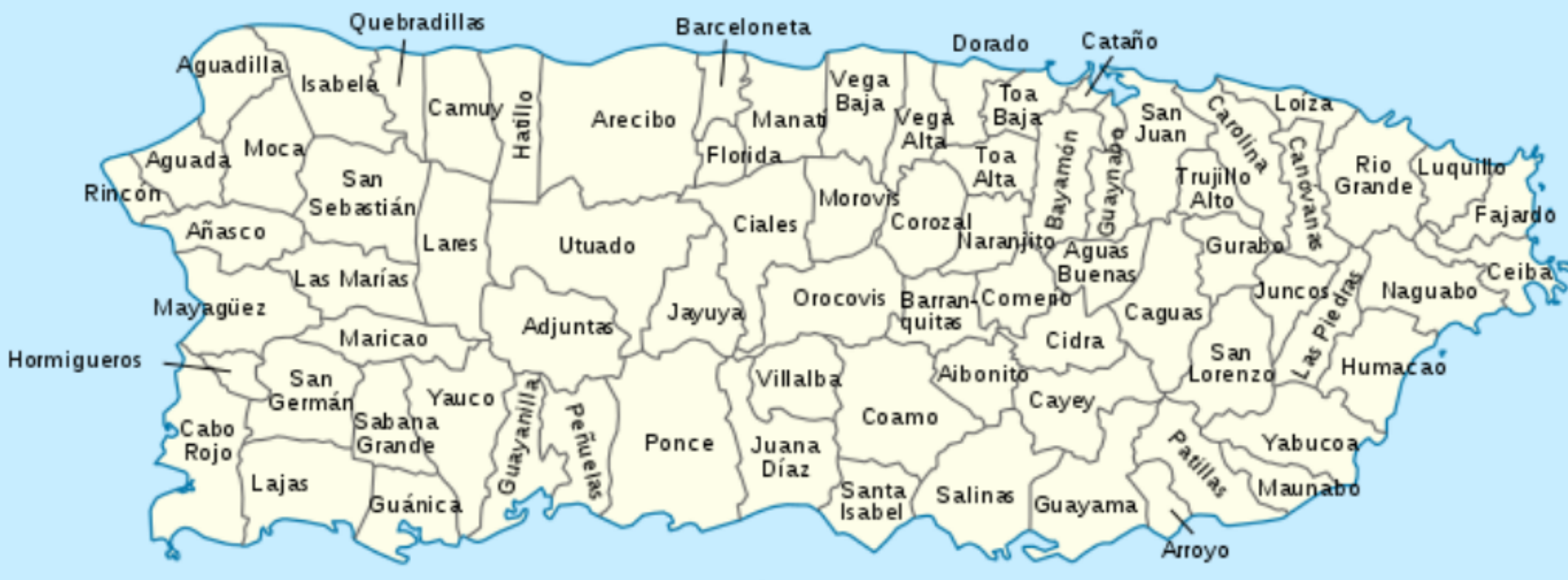}}
\subfigure
{\includegraphics[width=6.5in,height=2in]{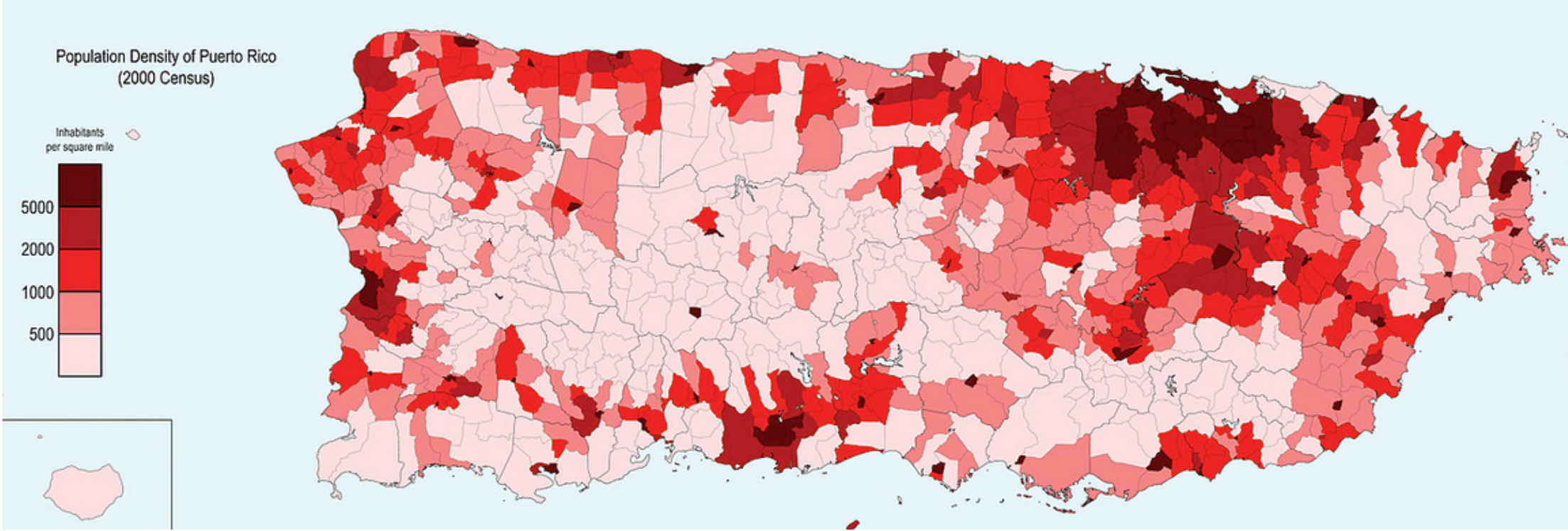}}
\caption{Top. The 76 municipalities in Puerto Rico (wikipedia.org).
Bottom. The population density of Puerto Rico (wikipedia.org)}
\label{Fig1a-1b}
\end{center}
\end{figure}

The initial susceptible population $S_0({\bf x})$ is obtained as follows: The boundary data of Puerto Rico are latitudes and longitudes obtained from Mathematica using CountryData["PuertoRico", "SchematicPolygon"], which forms a polygon with 71 points. The boundary data is used to generate the mesh. The population density data is obtained from 
http://sedac.ciesin.columbia.edu/data/collection/gpw-v4/maps/services, which gives the population in each 1 km $\times$1 km unit square area on earth. The latitude bounds are $\{17.9, 18.5\}$ and the longitude bounds are $\{-67.3, -65.3\}$, with 1216 mesh nodes. The population data is used to calculate the population density for the initial susceptible population $S_0({\bf x})$ based on 1216 mesh nodes (see Figure \ref{Fig2}).

\begin{figure}
\begin{center}
{\includegraphics[width=6.5in,height=3.0in]{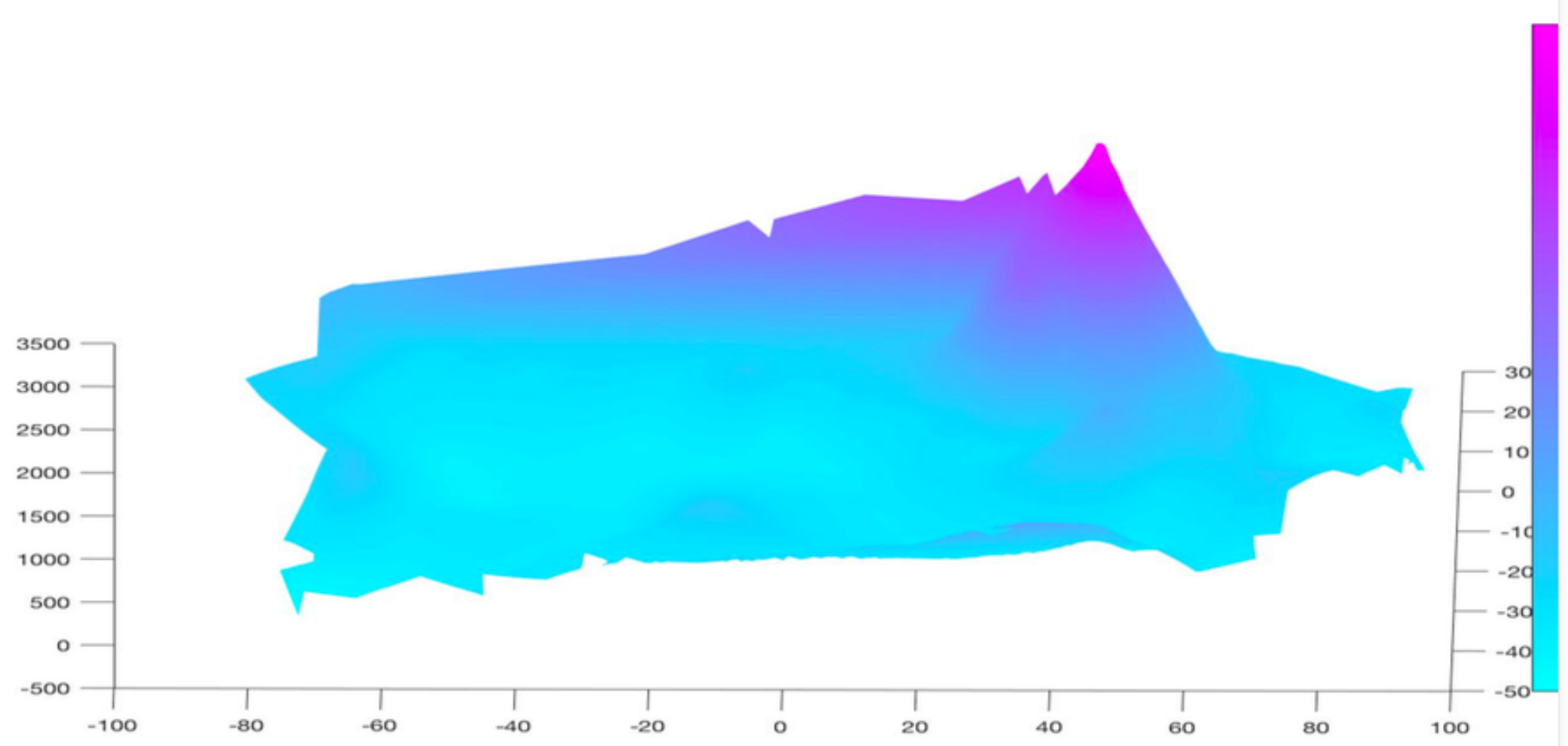}}
\caption{Top. The population density of the initial susceptible population 
$S_0({\bf x})$. }
\label{Fig2}
\end{center}
\end{figure}

\vspace{.8in}

\subsection{Parameterization of the Model for Puerto Rico}

The parameterization of any model of a seasonal influenza epidemic presents enormous challenges, because of the incompleteness of data. In the United States, typical epidemic data consists of Morbidity and Mortality Weekly Reports (MMWR) published by the Centers for Disease Control (CDC). For seasonal influenza, this data is very incomplete, and records only a small fraction of total cases. A recent analysis argued that unreported cases and attack rates (the fractions of the total susceptible populations that become infected over the course of an  epidemic) are largely underestimated \cite{Cauchemez2012}.

In an earlier study we developed a formalism for estimating the ratio of reported to unreported cases for the seasonal influenza  epidemics in Puerto Rico in 2015-2106 and 2016-9017 \cite{Magal2017}. The estimates in \cite{Magal2017} claimed attack rates of approximately 40\% to 50\%. These attack rates are higher than usually claimed for seasonal influenza epidemics. Here we have developed our parameters to reflect attack rates of approximated 30\% for both epidemics, based on a comparison of the graphs of the reported cases from CDC data and the graphs of the total cases (both reported and unreported) obtained from our model simulations. The objective was to match the duration of the epidemics, the turning points, and the character of their graphs in the reported case data and the model simulations.

Based on these considerations we estimate the parameters for Puerto Rico as follows:

\begin{enumerate} 
\item
Time units  are weeks. For the 2015-2016 epidemic, the initial time $t=0$ corresponds to week 44 of 2015. The 2015-2016 epidemic lasts approximately 30 weeks (until week 23 of 2016). For the 2016-2017 epidemic, the initial time $t=0$ corresponds to week 37 of 2016. The 2016-2017 epidemic lasts approximately 35 weeks (until week 21 of 2017) \cite{WebData}
\item
Spatial units are kilometers. The spatial region $\Omega$ is as in Figure \ref{Fig2a-2b}.
\item
The average length of the infectious period of infected people is about 2 days or 1/3.5 weeks: 
$\lambda({\bf x}) = 3.5.$ \cite{WebWiki}
\item
The transmission parameters for the nonlinear incidence form are $\tau({\bf x}) = 0.02, \, \kappa({\bf x}) = 0.05$, $p =1.0$, and $q = 1.0$. 
\item
The diffusion parameter of infected individuals is $\alpha = 4.0$, which corresponds indirectly to the geographical spread of the virus.
\end{enumerate}

In Figure \ref{Data} we graph the reported cases of influenza in Puerto Rico in 2015-2016 and 2016-2017, provided by Departamento de Salud, Gobierno de Puerto Rico, Sistema de Vigilancia de Influenza de Puerto Rico (\cite{WebData}). The seasonal influenza epidemics in Puerto Rico illustrate the importance of epidemic spatial heterogeneity. In 2015-2016 the graph of reported cases over 30 weeks showed 2 peaks: a high peak in week 10 and a low peak in week 25 (Figure \ref{Data}). In 2016-2017 the graph shows a low peak in week 7 and a high peak in week 12. Most models of disease transmission without spatial heterogeneity show only one peak (\cite{Magal2017}).  However, multiple peaks are possible if the spatial environment is heterogeneous, and we claim that spatial variation can explain the positions of the peaks.
\begin{figure}
\begin{center}
{\includegraphics[width=4.5in,height=3.5in]{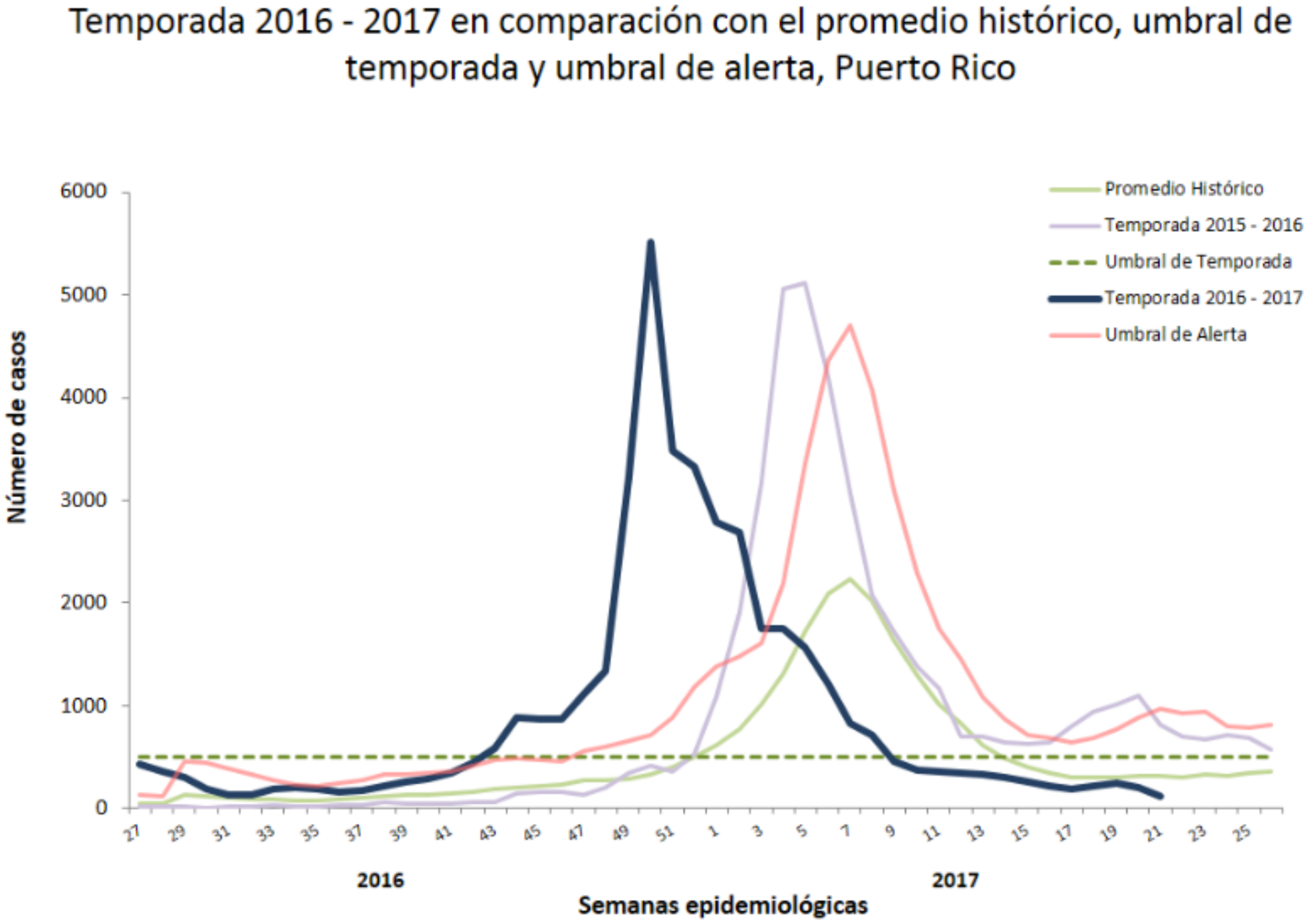}}
\end{center}
\caption{Reported  cases of seasonal influenza Puerto Rico in 2015-2016 (yellow graph) and 2016-2917 (black graph).}
\label{Data}
\end{figure}

\subsection{Simulation of the model for the 2015-2016 epidemic}
The model simulates the seasonal influenza epidemic for \underline{all} infected cases, not only reported cases. The simulation of all cases reflects the simulation of cases reported by Departamento de Salud, Puerto Rico (see Figure \ref{Data}). Estimates of the ratio of unreported to reported cases are difficult to obtain. For the US H1N1 epidemic in 2009, the CDC estimated this ratio as 79 -1 ({\it Health Day News}, October 29, 2009). The ratio for the simulation of the 2015-2016 epidemic is approximately 25-1. The estimated total infected cases for the 2015-2016 influenza epidemic are graphed in Figure \ref{total2015}, and our simulation Figure \ref{total2015} captures the feature of two peaks in the corresponding graph of reported case data  in Figure \ref{Data}.

\begin{figure}
\begin{center}
{\includegraphics[width=4.2in,height=2.1in]{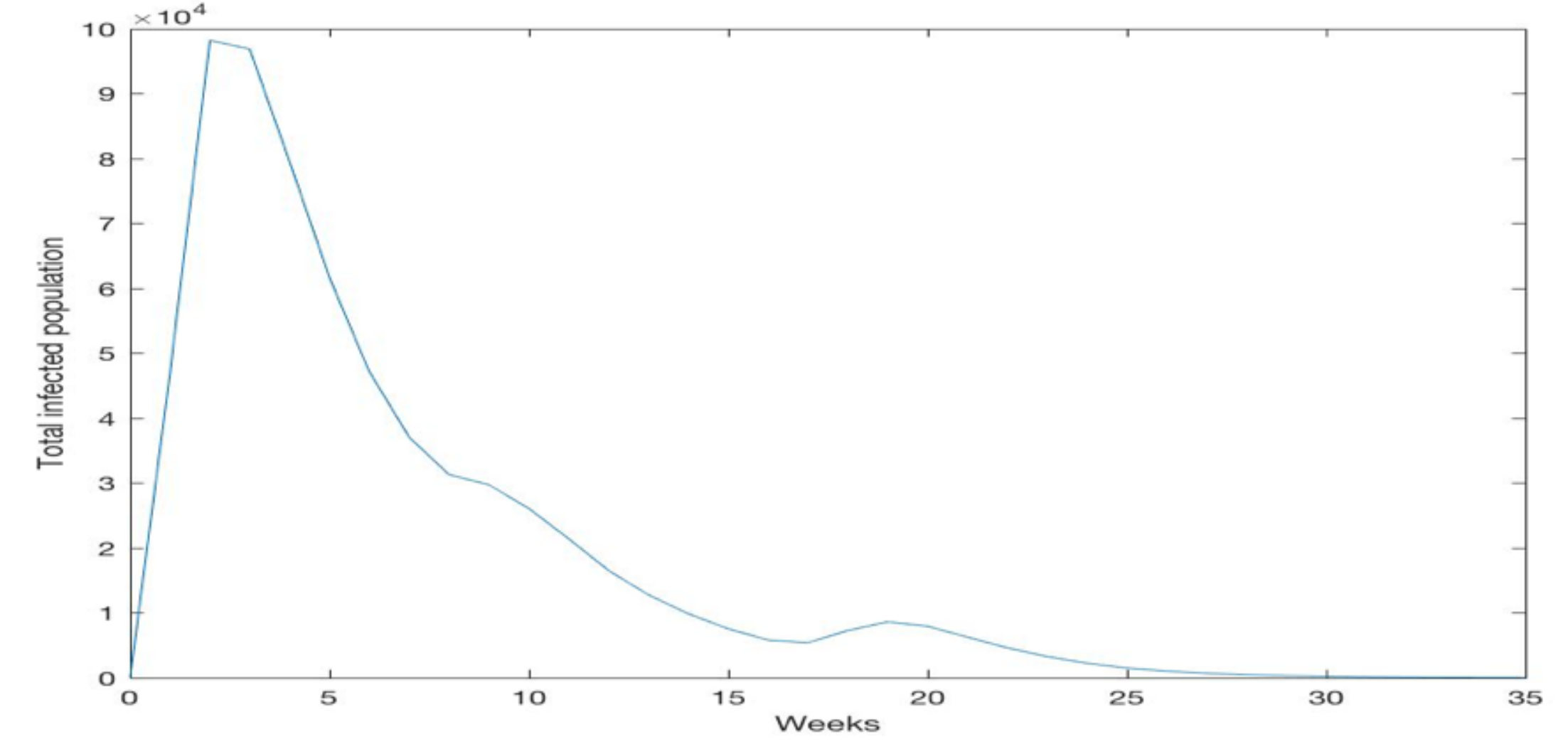}}
\end{center}
\caption{Model simulation of the total infected cases (including unreported cases) of the seasonal influenza 2015-2016 epidemic in Puerto Rico.}\label{total2015}
\end{figure}

We graph the density of the infected population at different times in Figure \ref{ModelDensities2015},
Figure \ref{2015-16-weeks2and6}, and Figure \ref{2015-16-weeks10and22}. From the location of the initial outbreak In San Juan, the epidemic spreads west toward  Arecibo, then south toward Ponce, and then west toward Mayag\H{u}ez. The two peaks in the total case count arise from the spatial evolution of the epidemic, first to the regions of San Juan (population 2,350,000) and Arecibo (population 193,000), and then to the regions of Ponce (population 262,000) and Mayag\H{u}ez (population 89,000). In Figure \ref{four2015} we graph the total infected cases in the four major municipalities of Puerto Rico). 

\begin{figure}
\begin{center}
{\includegraphics[width=6.5in,height=3.0in]{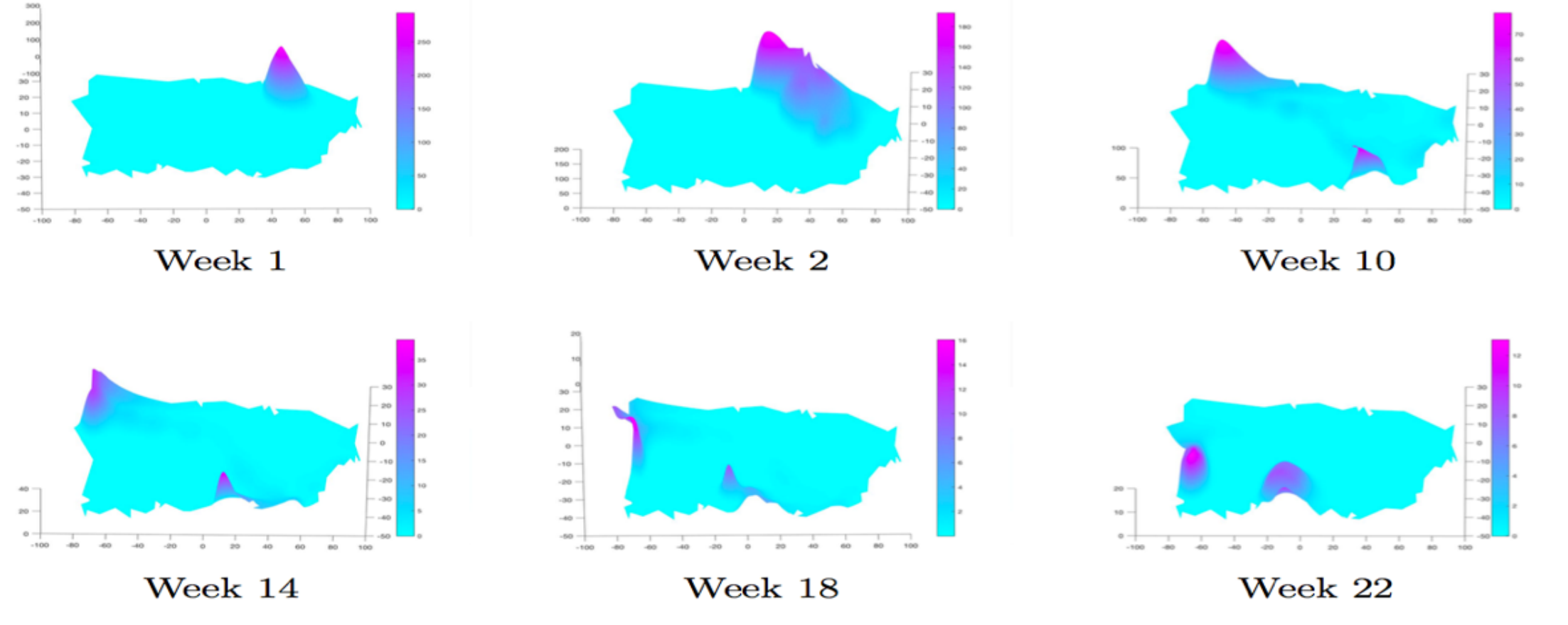}}
\caption{Model simulation for the spatial spread of the 2015-2016 seasonal influenza epidemic in Puerto Rico.}
\label{ModelDensities2015}
\end{center}
\end{figure}

\begin{figure}
\begin{center}
{\includegraphics[width=6.5in,height=3.0in]{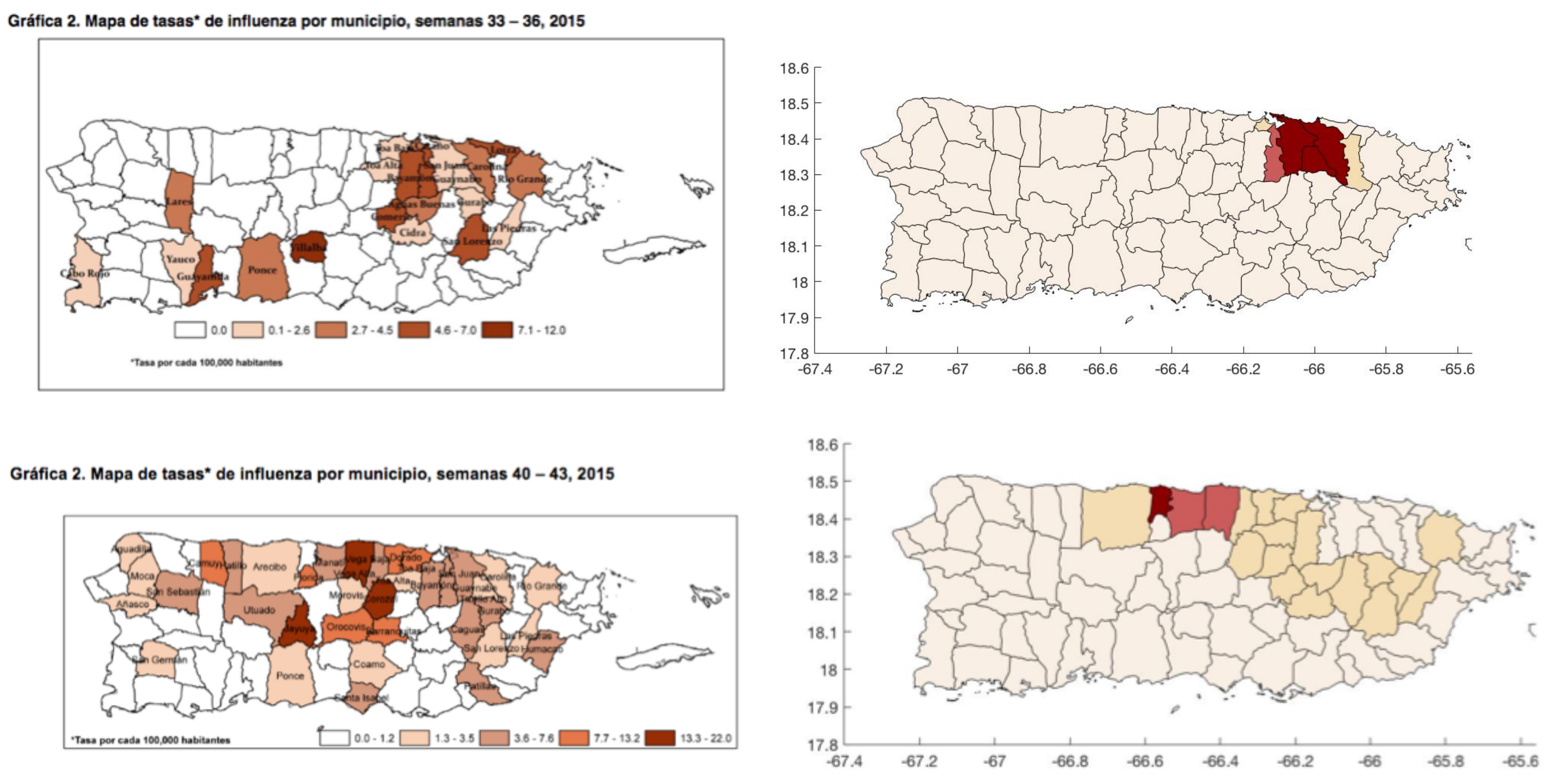}}
\caption{Infected population densities in the 2015-2016 seasonal influenza epidemic in Puerto Rico in all municipalities for week 2 (top) and week 6 (bottom) for data from Departamento de Salud (left) and the model simulation (right).}
\label{2015-16-weeks2and6}
\end{center}
\end{figure}

\begin{figure}
\begin{center}
{\includegraphics[width=6.5in,height=3.0in]{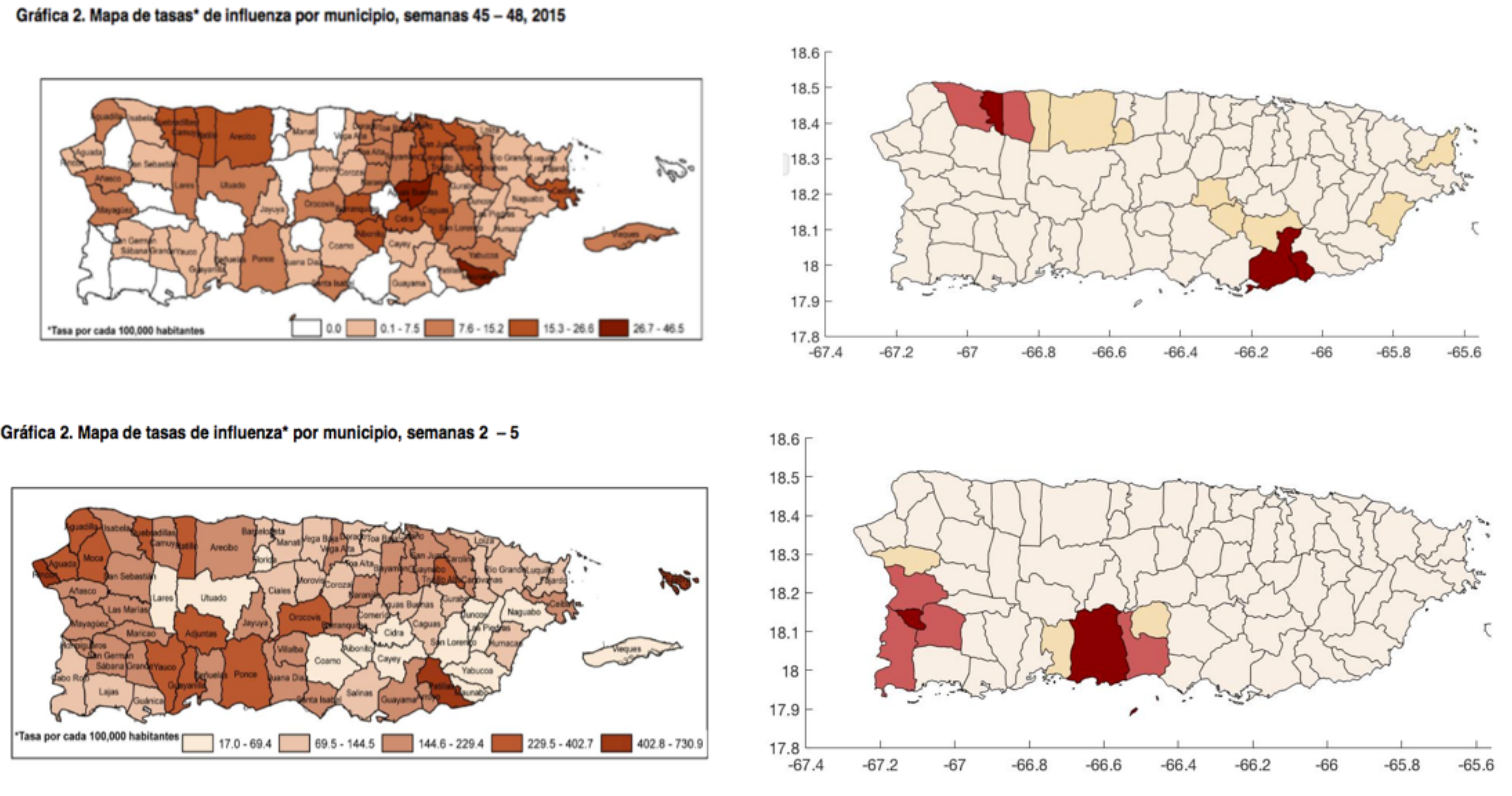}}
\caption{Infected population densities in the 2015-2016 seasonal influenza epidemic in Puerto Rico in all municipalities for week 10 (top) and week 22 (bottom) for data from Departamento de Salud (left) and the model simulation (right).}
\label{2015-16-weeks10and22}
\end{center}
\end{figure}

\begin{figure}
\centering \includegraphics[width=6.5in,height=3.3in]{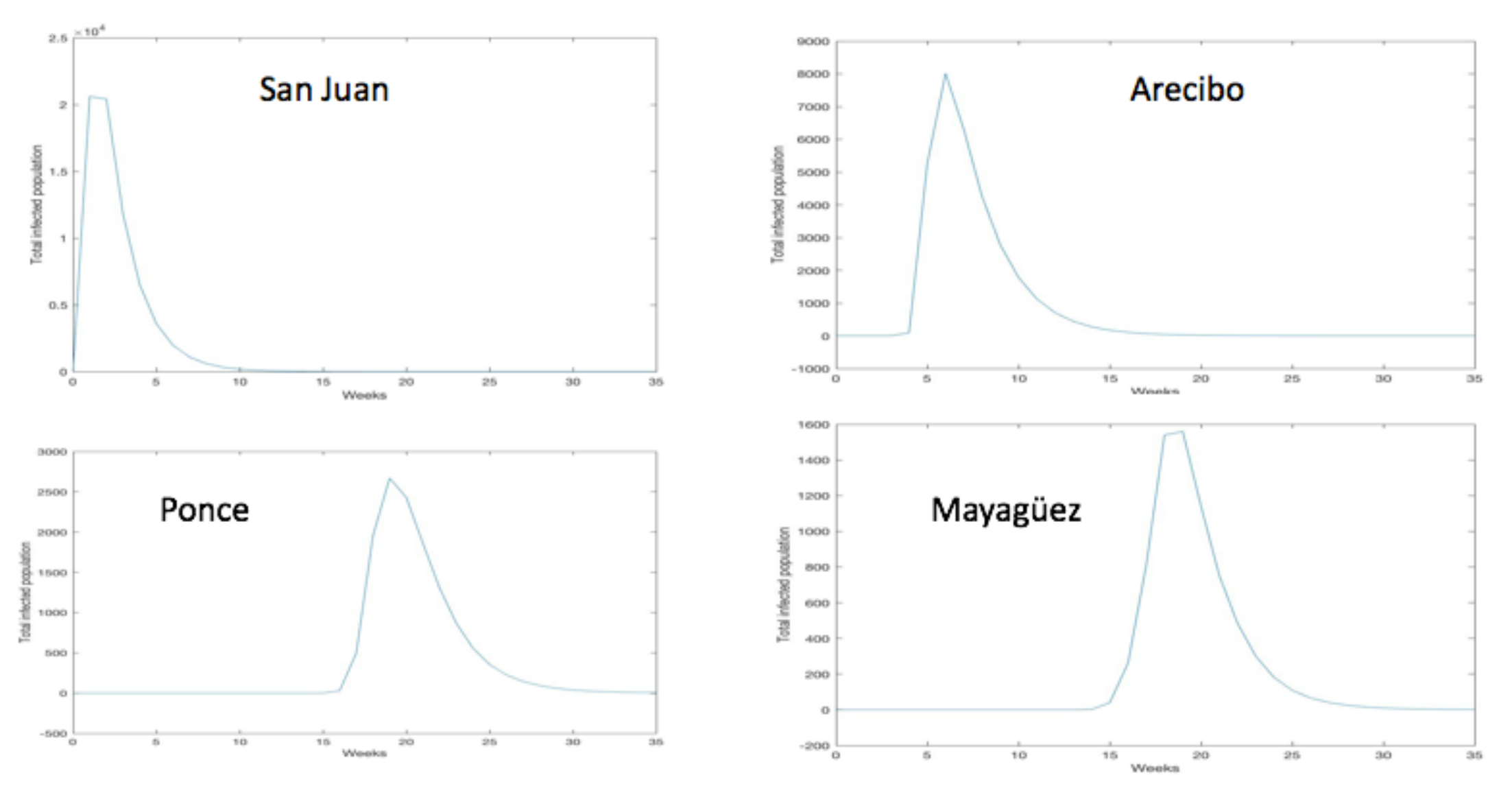}
\caption{The total infected cases in four major municipalities of Puerto Rico during the 2015-2016 influenza epidemic.}
\label{four2015}
\end{figure}

\subsection{Simulation of the model for the 2016-2017 epidemic}
A change in the outbreak location, with all other model inputs the same, simulates the data for the 2016-2017 influenza epidemic. We graph the total infected cases from the model simulation in Figure \ref{total2017}, and the graph (scaled) agrees with the graph for reported cases in Figure \ref{Data} for the 2016-2017 epidemic. We graph the density of the infected population at different time points  in Figure \ref{ModelDensities2017}. From the initial outbreak in Mayag\H{u}ez the epidemic spreads north and east toward Arecibo, then east toward San Juan, and south toward  Ponce. In Figure \ref{four2017} we graph the total infected cases in the four major municipalities of Puerto Rico.

\begin{figure}
\begin{center}
{\includegraphics[width=4.2in,height=2.1in]{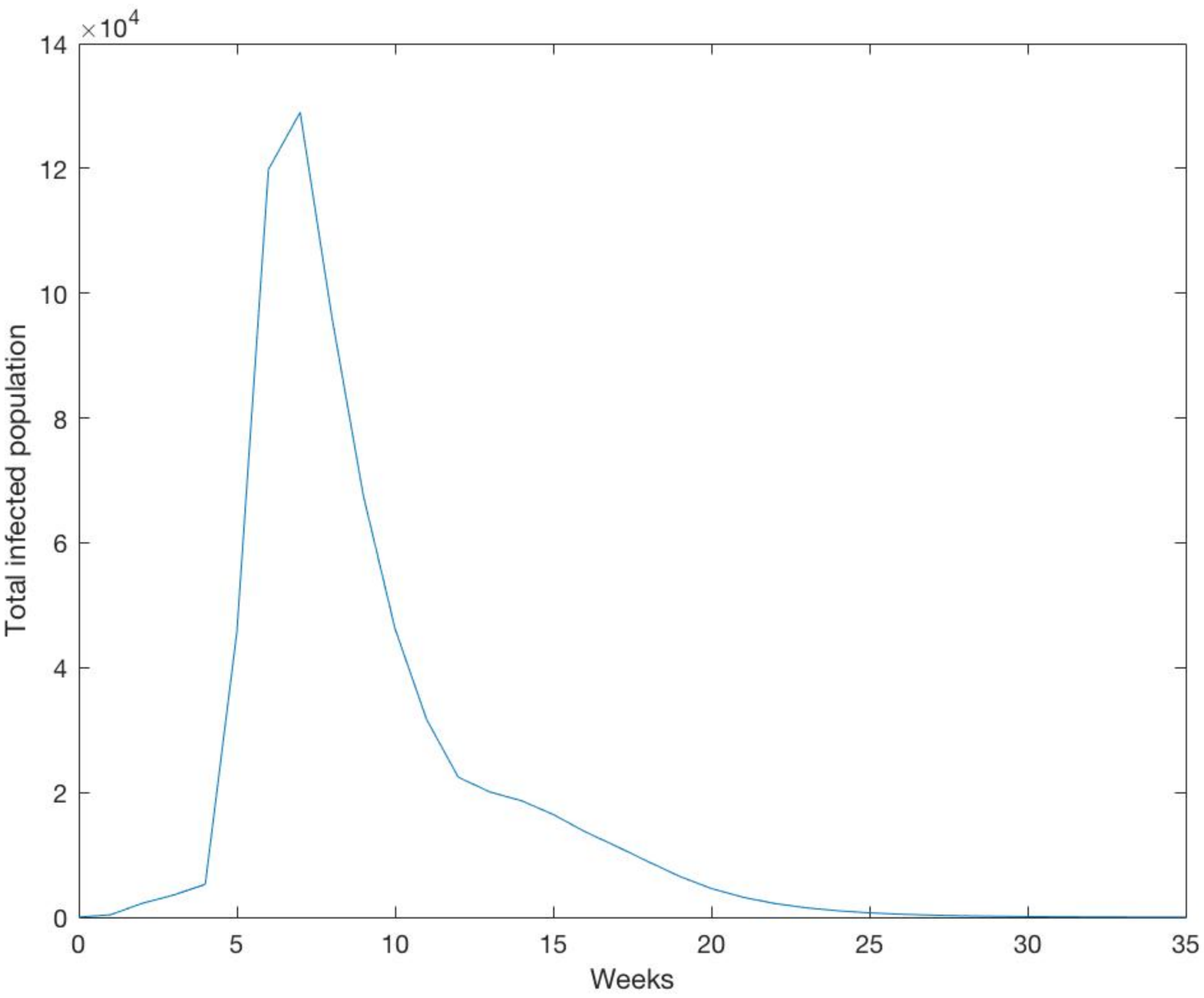}}
\end{center}
\caption{Model simulation of the total infected cases (including unreported cases) of the seasonal influenza 2016-2017 epidemic in Puerto Rico.}\label{total2017}
\end{figure}

\begin{figure}
\begin{center}
{\includegraphics[width=6.5in,height=3.0in]{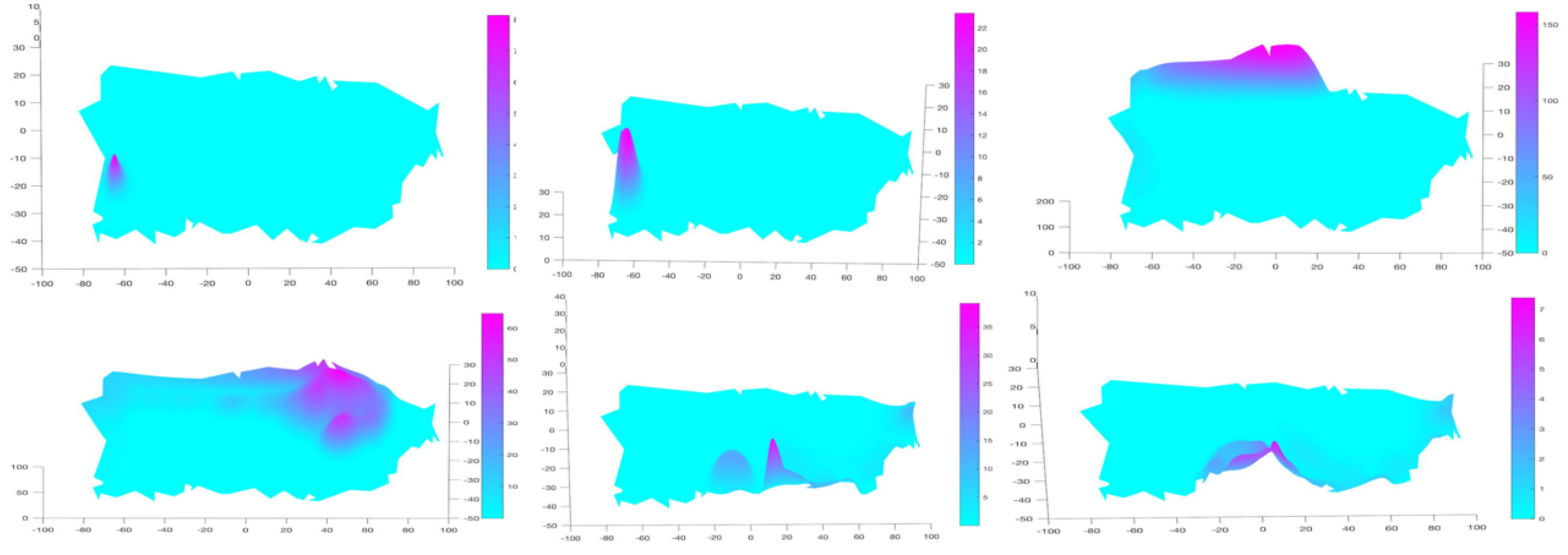}}
\caption{Model simulation for the spatial spread of the 2016-2017 seasonal influenza epidemic in Puerto Rico.}
\label{ModelDensities2017}
\end{center}
\end{figure}

\begin{figure}
\begin{center}
{\includegraphics[width=6.5in,height=3.0in]{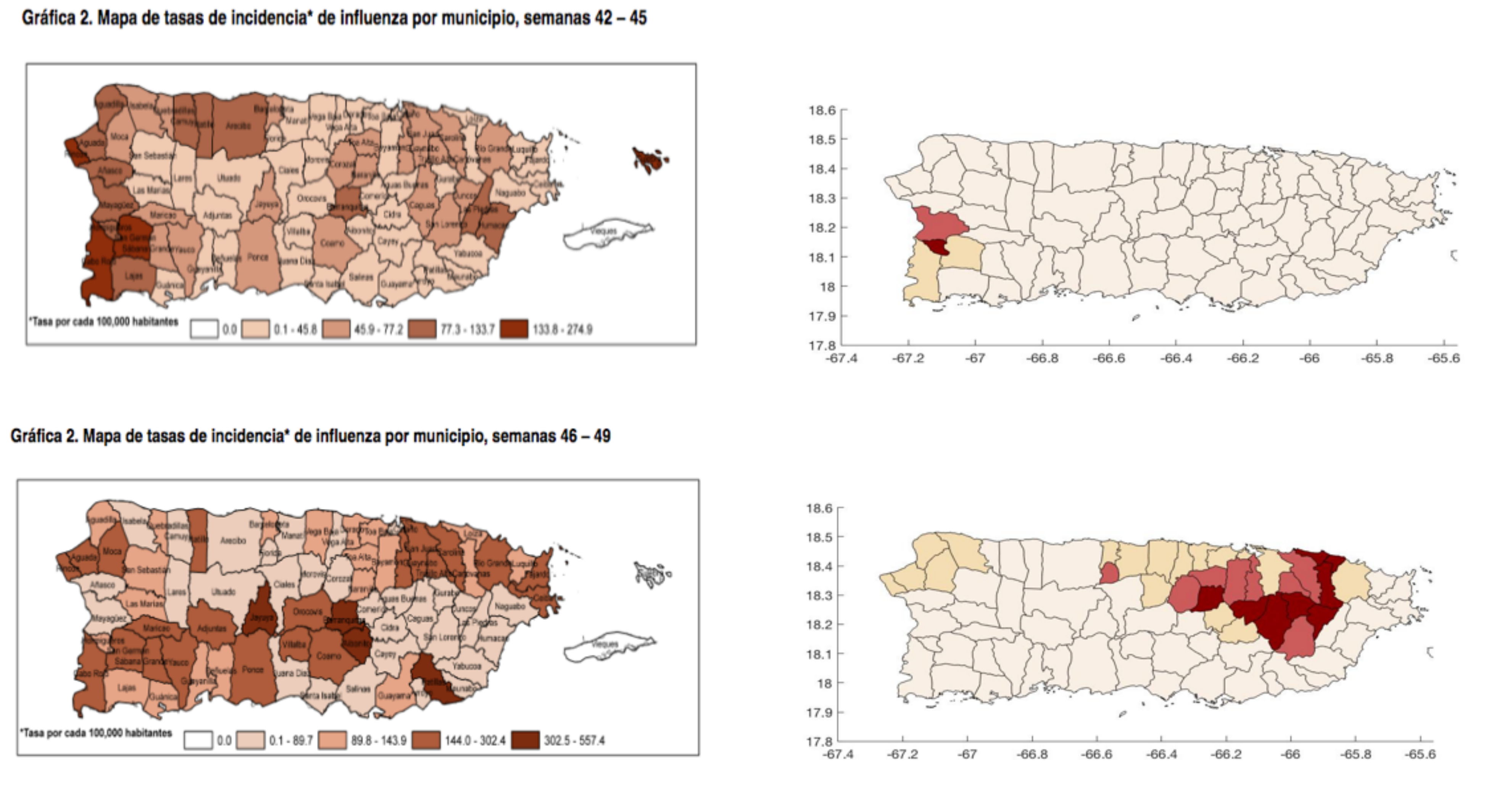}}
\caption{Infected population densities in the 2015-2016 seasonal influenza epidemic in Puerto Rico in all municipalities for week 4 (top) and week 10 (bottom) for data from Departamento de Salud (left) and the model simulation (right).}
\label{2016-17-weeks4and10}
\end{center}
\end{figure}

\begin{figure}
\begin{center}
{\includegraphics[width=6.5in,height=3.0in]{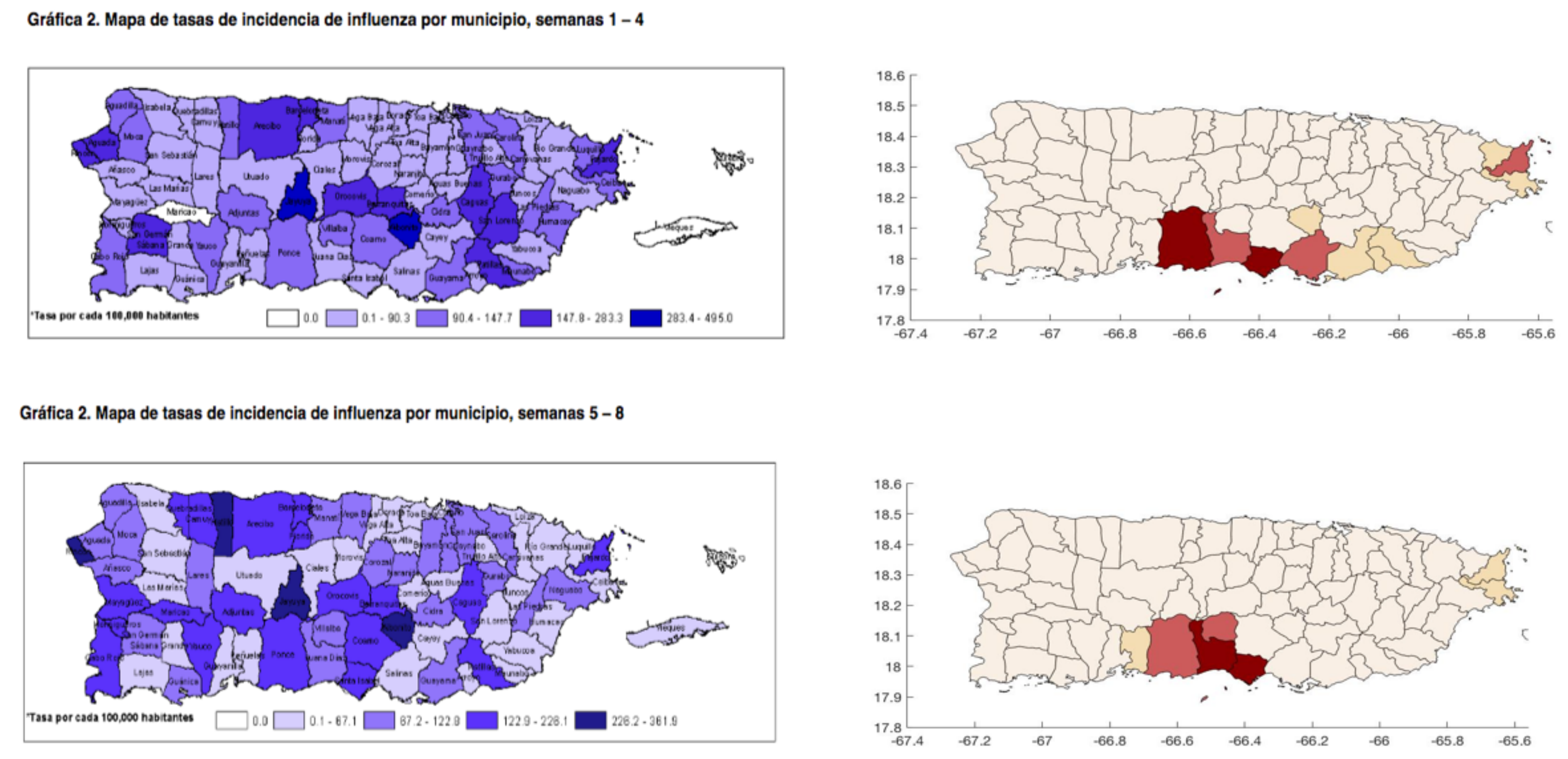}}
\caption{Infected population densities in the 2015-2016 seasonal influenza epidemic in Puerto Rico in all municipalities for week 18 (top) and week 22 (bottom) for data from Departamento de Salud (left) and the model simulation (right).}
\label{2016-17-weeks18and22}
\end{center}
\end{figure}

\begin{figure}
\centering \includegraphics[width=6.5in,height=3.3in]{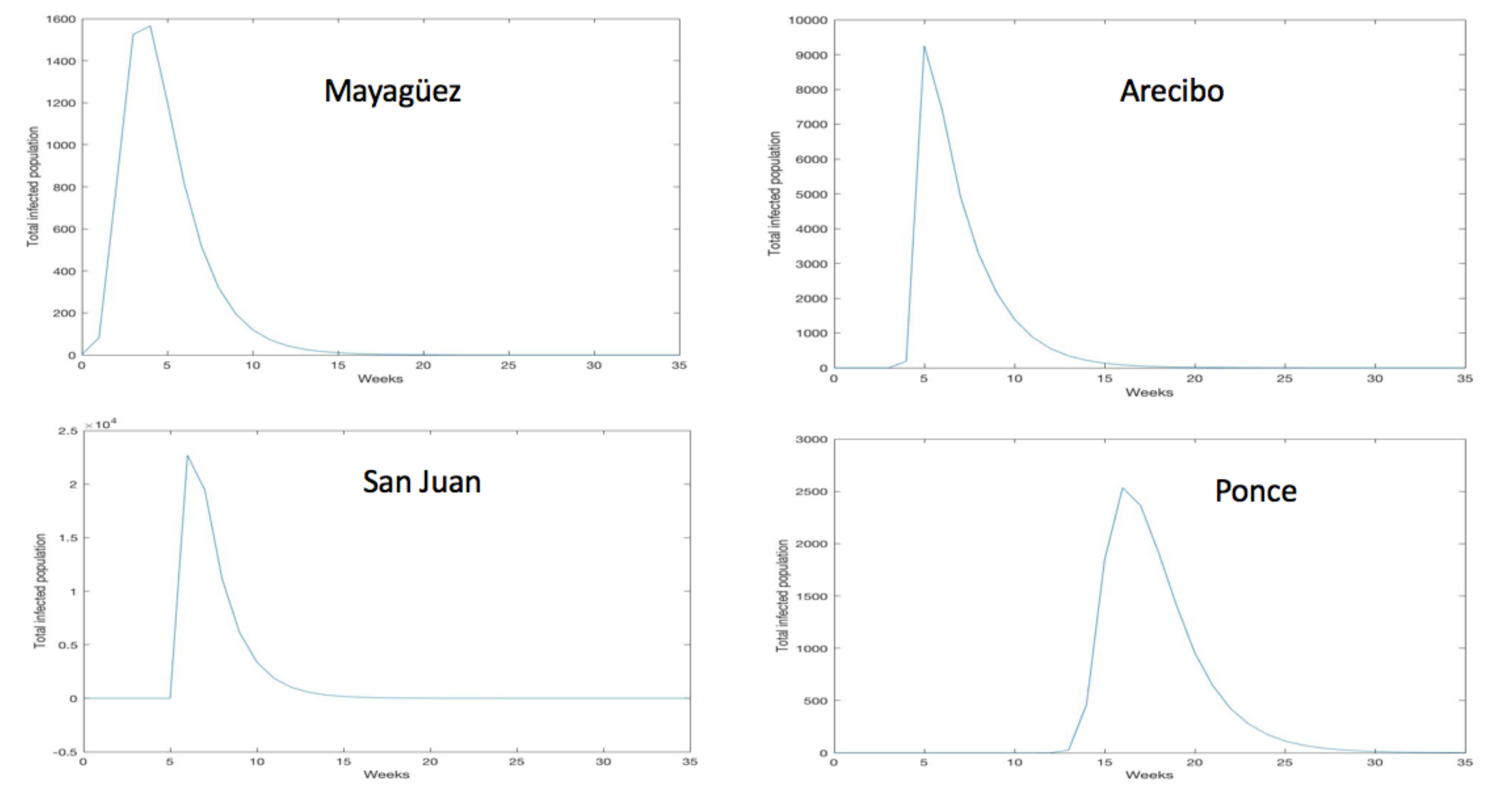}
\caption{The total infected cases in four major municipalities of Puerto Rico during the 2016-2017 influenza epidemic.}
\label{four2017}
\end{figure}

We introduce the local basic reproduction number
  $$
  R_0(x,y) = \frac{\tau S_0(x,y) I_0(x,y)^{p-1}}{\lambda(x,y) (1+\kappa I_0(x,y))^q}.
  $$
  The local evolution of the epidemic at a given outbreak location $(x,y)$ is governed by the local basic reproduction number.  If $R_0(x,y) < 1$, the epidemic initially subsides, then grows. If $R_0(x,y) > 1$, the epidemic initially expands. 

\section{Conclusions and discussion}
The model indicates that influenza in Puerto Rico rises each season from initial small outbreak locations, and spreads through most of the island, dependent on geographic population variation. The final size of the epidemic at the end of the season depends on the initial outbreak locations, the geographic heterogeneity of the population, and the model parameters.

The model suggests a reason for the seasonality of seasonal influenza epidemics. In a general region, the epidemic lasts approximately 30 weeks, but in subregions the epidemic last approximately 6 weeks (although sometimes re-occurring). The model indicates that the epidemic duration depends strongly on the depletion of the susceptible population to a level that no longer sustains transmission. This depletion happens rapidly in local regions, while the general level of the epidemic occurs much longer in larger regions. Thus, geographic variation is important in understanding the seasonality of seasonal influenza epidemics.

The model indicates that the most effective controls are to monitor the importation of infected people into local regions, and to concentrate public health interventions in regions of high population density (where the local basic reproduction number $R_0(x,y)$ is highest), especially at the beginning of the season.

Future work involves the use of disease age to track infectiousness levels of infected individuals, through the incubation period, and the rise and fall of the infectious period. Particular emphasis will be given to pre-symptomatic infectiousness periods. The model will be extended to include public policy measures such as quarantine, vaccination, and school closings. Future work will extend the model to study geographic variation in other diseases, including vector-borne diseases such as zika, dengue, and malaria.

\begin{appendices}
\begin{theorem}\label{ode} Let $\tau, \kappa, \lambda, p, q > 0$ with $1 \leq p \leq q+1$, and let $S_0, I_0 >0$. 
There exists a unique solution  $S(t) \geq 0$, $I(t) \geq 0$,  satisfying $S(0) = S_0$,  $I(0) = I_0$ and
\begin{eqnarray}
S^{\prime}(t) &=&  
  - \frac{ \tau \, I(t)^p}{1+ \kappa \,  I(t)^q} S(t), \hspace{1.8in}  (ODE.1) \nonumber \\  
I^{\prime}(t) &=&  \frac{ \tau \, I(t)^p}{1+ \kappa \,  I(t)^q} S(t)- \lambda I(t). \nonumber  \hspace{1.38in} (ODE.2)   \nonumber 
\end{eqnarray}
Let $R_0 = \tau I_0^{p-1} S_0/(\lambda (1 + \kappa I_0^q)$. If $R_0 < 1$, then $S(t)$ decreases to a limiting value $S_{\infty} > 0$ and $I(t)$ decreases to $0$. If $R_0 > 1$, then $S(t)$ decreases to a limiting value $S_{\infty} > 0$ and $I(t)$ first increases, then decreases to $0$.
\end{theorem}
\begin{remark}
 For the spatially independent case, the graph of $I(t)$ can have at most one peak.
\end{remark}
\begin{proof}
Add (ODE.1) to (ODE.2) and integrate it over $(0, t)$ to obtain
\begin{equation}
\label{eqA1}
0 \leq S(t) + I(t) + \lambda \int_0^t I(t) dt  = S_0 + I_0.
\end{equation}
The existence of a unique nonnegative solution on $[0,\infty)$ follows from standard theory. 
Since $S^{\prime}(t) \leq 0$, $S(t)$  converges to a limt $S_{\infty} \geq 0$.
Also, $S(t)$  and $I(t)$ are bounded on $[0,\infty)$, $I^{\prime}(t)$ is bounded on $[0,\infty)$, and
$$\int_0^{\infty} I(t) dt  < \infty,$$
which implies that $\lim_{t \rightarrow \infty} I(t) = 0$.

Noticing $1\le p\le q+1$, a simple calculation shows that
$$\frac{z^p}{1 + \kappa \, z^q}  \leq max(1,1/\kappa) z , \, z \geq 0.$$
Thus,
\begin{equation}
\label{eqA2}
\int_0^{\infty}  \frac{ \tau \, I(t)^p}{1+ \kappa \, I(t)^q} dt \leq \tau  max(1,1/\kappa)\,\int_0^{\infty} I(t) dt \, <  \, \infty.
\end{equation}
Divide both sides of (ODE.1) by $S(t)$ and integrate it over $(0, t)$ to obtain 
$$ log\bigg(\frac{S(t)}{S_0} \bigg)= - \int_0^t \frac{ \tau \, I(s)^p}{1+ \kappa \,  I(t)^q} ds$$
which implies 
$$S_{\infty}  = S_0 \, Exp\bigg(- \int_0^{\infty}  \frac{ \tau \, I(t)^p}{1+ \kappa \, I(t)^q} dt \bigg) \ne 0.$$

Then to show that $I(t)$ can have at most one peak, observe from (ODE.2) 
$$I^{\prime \prime}(t) = \Bigg( \bigg(1+\kappa I(t)^q \bigg) \bigg(\tau p I(t)^{p-1} I^{\prime}(t) S(t) + \tau I(t)^p S^{\prime}(t) \bigg)$$
$$ - \, \, \bigg(\tau I(t)^p S(t) \bigg) \bigg( \kappa q I(t)^{q-1} I^{\prime}(t) \bigg) \Bigg) \Bigg/  \bigg(1 + \kappa I(t)^q \bigg)^2-  \lambda I^{\prime}(t).$$
If $I^{\prime}(\bar{t}) = 0$, then 
$$I^{\prime \prime}(\bar{t}) =\frac{\tau I(\bar{t})^p S^{\prime}(\bar{t})}
{1+\kappa I(\bar{t})^q} < 0,$$
which implies $I(t)$ is concave down wherever $I^{\prime}(\bar{t}) = 0$.

Rewrite (ODE.2) as 
$$
I'(t)=\lambda \left(\frac{ \tau \, I(t)^{p-1}S(t)}{1+ \kappa \,  I(t)^q} - 1\right)I(t).
$$
Then we can see that $I(t)$ decreases at $t=0$ if $R_0<1$ and increases at $t=0$ if $R_0>1$. So the claim on $I(t)$ follows from the fact that $I(t)$ converges to zero and has at most one peak. 
\end{proof}

\begin{remark}
 If $ p=1$ and $\kappa = 0$, one can combine \eqref{eqA1} and \eqref{eqA2} to obtain
$$S_{\infty}  +  \frac{\lambda}{\tau} log\bigg(\frac{S_{\infty}}{S_0}\bigg) = S_0 + I_0.$$
\end{remark}

\begin{theorem}\label{pde}
Let $\Omega$ be a bounded domain in $R^n$ with smooth boundary $\partial \Omega$. 
Let $\alpha, \tau, \kappa, p, q$ be positive constants with  $1 \leq p \leq q+1$, 
let $\lambda   \in C_+(\overline{\Omega})$ with $\lambda({\bf x}) \geq \lambda_0 > 0$ for all ${\bf x} \in \Omega$, and let 
$S_0,  \,I_0  \in L_+^1(\overline{\Omega})$ be nontrivial.

Then there exists unique 
$S(t,\cdot), I(t, \cdot): [0,\infty)\rightarrow L_{+}^{1}(\overline{\Omega})$ satisfying
\begin{eqnarray}
\frac{\partial}{\partial t} S(t,{\bf x})  &=&  
  - \frac{ \tau \,  I(t,{\bf x})^p}{1+ \kappa \,  I(t,{\bf x})^q} S(t,{\bf x}), \, {\bf x} \in \Omega, \, t >0 \hspace{2in}  (PDE.1) \nonumber \\  
\frac{\partial}{\partial t} I(t,{\bf x})  &=& \alpha \Delta I(t,{\bf x})
  + \frac{ \tau \,  I(t,{\bf x})^p}{1+ \kappa \,  I(t,{\bf x})^q}  S(t,{\bf x}) - \lambda({\bf x}) I(t,{\bf x}), \, {\bf x} \in  \Omega, \, t >0 \hspace{.25in} (PDE.2) \nonumber \\  
\frac{\partial}{\partial \eta} I(t,{\bf x}) &=& 0, \, {\bf x} \in \partial \Omega, \, t >0 \hspace{3.3in} (PDE.3) \nonumber \\  
S(0,{\bf x}) &=& S_0({\bf x}),\, \, I(0,{\bf x}) = I_0({\bf x}), \, {\bf x} \in  \Omega. \hspace{2.32in} (PDE.4) \nonumber 
\end{eqnarray}
Further, $\lim_{t \rightarrow \infty}S(t,\cdot) = S_{\infty}(\cdot) \geq 0, \, \lim_{t \rightarrow \infty}I(t,\cdot) = 0$ in $L^1(\Omega)$, and  $S_{\infty}(\cdot)  \ne 0$.
\end{theorem}

\begin{proof}
Add (PDE.1) and (PDE.2) and integrate it  over $(0, t)$ and $\Omega$ to obtain
$$\iint_{\Omega}(S(t,{\bf x}) + I(t,{\bf x})) d{\bf x} + \iint_{\Omega} \bigg( \int_0^t \lambda({\bf x}) I(s,{\bf x}) ds \bigg) d{\bf x} =
\iint_{\Omega}(S_0({\bf x}) + I_0({\bf x}))d{\bf x},$$
which implies
$$
 \iint_{\Omega} \bigg( \int_0^\infty \lambda({\bf x}) I(t,{\bf x}) dt \bigg) d{\bf x} <\infty.
$$

As in the ODE case, since $\lambda({\bf x}) \geq \lambda_0 > 0$,
$$\iint_{\Omega} \bigg(\int_0^t \frac{ \tau \, I(t,{\bf x})^p}{1+ \kappa \, I(t,{\bf x})^q} dt \bigg) d{\bf x} \leq \tau max(1,1/\kappa)
\iint_{\Omega}  \bigg(\int_0^t I(t,x) dt \bigg)dx <  \infty.$$
The existence of a unique nonnegative solution in $L^1(\Omega)$ on $[0,\infty)$ follows from standard theory. 
As in the ODE case, (PDE.1) implies that for a.e. ${\bf x} \in \Omega, \,  \frac{\partial}{\partial t}S(t,{\bf x}) \leq 0$ and 
$\lim_{t \rightarrow \infty}S(t,{\bf x}) = S_{\infty}({\bf x}) \geq 0$.
By the Lebesgue Theorem $\lim_{t \rightarrow \infty}S(t,\cdot) = S_{\infty}(\cdot)$ in $L^1(\Omega)$.
Integrate (PDE.1) over $t$ to obtain for a.e.  ${\bf x} \in \Omega$,
$$ log\bigg(\frac{S(t,{\bf x})}{S_0({\bf x})} \bigg)= - \int_0^t \frac{ \tau \, I(s,{\bf x})^p}{1+ \kappa \,  I(s,{\bf x})^q} ds.$$
Then
$$\iint_{\Omega} \bigg(\int_0^{\infty} \frac{ \tau \, I(t,{\bf x})^p}{1+ \kappa \, I(t,{\bf x})^q} dt \bigg) dx< \infty \Rightarrow 
\int_0^{\infty}  \frac{ \tau \, I(t,{\bf x})^p}{1+ \kappa \, I(t,{\bf x})^q} dt < \infty \mbox{ a.e. } {\bf x} \in \Omega.$$
$$\mbox{Thus,} \, S_{\infty} \ne 0, \, \mbox{since} \, S_{\infty}({\bf x}) = S_0({\bf x}) \, Exp\bigg(- \int_0^{\infty}  \frac{ \tau \, I(t,{\bf x})^p}{1+ \kappa \, I(t,{\bf x})^q} dt \bigg) \mbox{ for a.e. } {\bf x} \in \Omega. \hspace{1.0in} $$

For $S_0,I_0 \in L_+^1(\Omega)$, define the $\omega$-limiting set of $(S_0,I_0)$  in $[L^1(\Omega)]^2$ as  
$$\{(u, v)\in [L^1(\Omega)]^2: (S(t_n,\cdot),I(t_n,\cdot))\rightarrow (u, v)  \text{ in }  [L^1(\Omega)]^2 \text{ for some } \{t_n\}\}.$$
The $\omega$-limiting set of $(S_0,I_0)$ is bounded in $L^1(\Omega)$.
Since $S(t,\cdot)$ is convergent in $L^1({\Omega})$,
$\{S(t,\cdot) : t \geq 0 \}$ is compact in $L^1(\Omega)$.
Since the linear operator semigroup  generated by the Laplacian with Neumann boundary conditions is 
compact  in $L^1(\Omega)$, the nonlinear term in (PDE.2) is bounded in $t$, and since $\lambda_0 > 0$, $\{I(t,\cdot) : t \geq 0 \}$  has compact closure in $L^1(\Omega)$.
Thus, the $\omega$-limiting set of $(S_0,I_0)$ is non-empty in $L^1(\Omega)$.

To prove $\lim_{t \rightarrow \infty} I(t) = 0$, define 
$V(S,I)(t) = \iint_{\Omega}(S(t,{\bf x}) + I(t,{\bf x})) d{\bf x}$ and add (PDE.1) and (PDE.2) to obtain
$$\dot{V}(S,I)(t) = - \iint_{\Omega} \bigg( \int_0^t \lambda({\bf x}) I(t,{\bf x}) \bigg) d{\bf x} \leq 0.$$
By the Invariance Principle $(S(t), I(t))$ converges to $(S_{\infty},0)$ in $L^1(\Omega)$, 
since the maximal invariant subset of $\{ (S,I): \dot{V}(S,I)= 0\}$ in $L^1(\Omega)$ is $(S_{\infty},0). \square$
\end{proof}

\end{appendices}


\begin{thebibliography}{10}

\bibitem{WebWiki} \url{https://en.wikipedia.org/wiki/Influenza}.

\bibitem{WebData} \url{http://www.salud.gov.pr/Estadisticas-Registros-y-Publicaciones/Pages/Influenza.aspx}.

\bibitem{balcan2009multiscale} D. Balcan, V. Colizza, B. Gon{\c{c}}alves, H. Hu, J.J Ramasco, and A. Vespignani, Multiscale mobility networks and the spatial spreading of infectious diseases, {\it Proceedings of the National Academy of Sciences},
  \textbf{106(51)} (2009), 21484-21489.
  
\bibitem{Biggerstaff} M. Biggerstaff and L. Balluz, 
Self-reported influenza-like illness during the 2009 H1N1 influenza pandemic, \textit{United States, Morbid Mortal Weekly Rep., September 2009ÐMarch 2010}, \textbf{60} (2011), 60:37.

\bibitem{capasso1978global} V. Capasso, Global solution for a diffusive nonlinear deterministic epidemic model, {\it  SIAM Journal on Applied Mathematics}, \textbf{35(2)} (1978), 274-284.

\bibitem{Cauchemez2012} S. Cauchemez, P. Horby, A. Fox, {\it et al.}, Influenza infection rates, measurement errors and the interpretation of paired serology, {\it PLOS Pathogens} (2012).

\bibitem{colizza2007modeling} V. Colizza, A. Barrat, M. Barthelemy, A.-J. Valleron and A. Vespignani, Modeling the worldwide spread of pandemic influenza: baseline case and containment interventions, {\it PLoS medicine}, \textbf{4(1)} (2007), e13.

\bibitem{eubank2004modelling} S. Eubank, H. Guclu, V.S.A. Kumar, M.V. Marathe, Modelling disease outbreaks in realistic urban social networks, {\it Nature}, \textbf{429(6988)} (2004), 180.

\bibitem{ferguson2005strategies} N.M. Ferguson, D.A. Cummings, S. Cauchemez and C. Fraser, Strategies for containing an emerging influenza pandemic in southeast asia, \newblock {\it Nature}, \textbf{437(7056)} (2005), 209.

\bibitem{fitzgibbon2008simple} W.E. Fitzgibbon and M. Langlais, Simple models for the transmission of microparasites between host populations living on noncoincident spatial domains. In  P. Magal and S. Ruan (Eds.) {\em Structured population models in biology and epidemiology}, pages 115-164. Springer, 2008.

\bibitem{fitzgibbon2017outbreak} W.E. Fitzgibbon, J.J. Morgan, and G.F. Webb, An outbreak vector-host epidemic model with spatial structure: the 2015-2016 zika outbreak in rio de janeiro, {\it Theoretical Biology and Medical Modelling}, \textbf{14(1)} (2017).

\bibitem{germann2006mitigation} T.C. Germann, K. Kadau, I.M. Longini, and C.A. Macken, Mitigation strategies for pandemic influenza in the united states, {\it Proceedings of the National Academy of Sciences}, \textbf{103(15)} (2006), 5935-5940.

\bibitem{gog2014} J.R. Gog, S. Ballesteros, C. Viboud, L. Simonsen, {\it et al.}, Spatial transmission of 2009 pandemic influenza in the US, {\it PLOS Computational Biology}, \textbf{10(6)} (2014).

\bibitem{grais2003assessing} R.F. Grais, J.H. Ellis, and G.E. Glass, Assessing the impact of airline travel on the geographic spread of
  pandemic influenza, {\it European journal of epidemiology}, \textbf{18(11)} (2003), 1065-1072.
  
\bibitem{hethcote1991} H.W. Hethcote and  P. van den Driessche,
Some epidemiological models with nonlinear incidence,
J. Math. Biol., \textbf{29} (1991), 271-287.

\bibitem{hufnagel2004forecast} L. Hufnagel, D. Brockmann and T. Geisel, Forecast and control of epidemics in a globalized world, {\it Proceedings of the National Academy of Sciences of the United States of America}, \textbf{101(42)} (2004), 15124-15129.

\bibitem{liu1987} W.M. Liu, H.W. Hethcote, and S.A. Levin,
Dynamical behavior of epidemiological models with nonlinear incidence rates,
J. Math. Biol., \textbf{25} (1987), 359-380.

\bibitem{longini2005containing} I.M. Longini, A. Nizam, S. Xu, K. Ungchusak, W. Hanshaoworakul, D.A. Cummings and  M.E. Halloran, Containing pandemic influenza at the source, {\it Science}, \textbf{309(5737)} (2005),  1083-1087.

\bibitem{Magal2017} P. Magal and G.F. Webb, The parameter identification problem for SIR epidemic models: Identifying unreported cases, preprint.

\bibitem{merler2010role} S. Merler and M. Ajelli, The role of population heterogeneity and human mobility in the spread of pandemic influenza, \textit{Proceedings of the Royal Society of London B: Biological Sciences}, \textbf{277(1681)} (2010), 557-565.

\bibitem{rass2003spatial} L. Rass and J. Radcliffe, {\it Spatial deterministic epidemics}, Vol. 102. American Mathematical Soc., 2003.

\bibitem{Reed} C. Reed, F.J. Angulo, D.L. Swerdlow, M. Lipsitch, M.I. Meltzer, D. Jernigan, and L. Finelli, 
Estimates of the prevalence of pandemic (H1N1) 2009, United States, AprilÐJuly 2009,
\textit{Emer. Inf. Dis.}, 
\textbf{15, No.12} (2009).

\bibitem{ruan2003} S. Ruan and W. Wang, Dynamical behavior of an epidemic model with a nonlinear incidence rate, {\textit J. Dif. Eqs}, \textbf{1(10)} (2003), 135-163.

\bibitem{ruan2007spatial} S. Ruan, Spatial-temporal dynamics in non local epidemiological models, In Y. Takeuchi, Y. Iwasa, K. Sato (Eds.)  {\it Mathematics for life science and medicine}, Springer Berlin Heidelberg, (2007), 97-122.

\bibitem{rvachev1985mathematical} L.A. Rvachev and I.M. Longini, A mathematical model for the global spread of influenza, {\it Mathematical biosciences}, \textbf{75(1)} (1985), 3-22.

\bibitem{Shrestha} S.S. Shrestha, D.L. Swerdlow, R.H. Borse, V.S. Prabhu, L. Finelli, C.Y. Atkins, K. Owusu-Edusei, B. Bell, P.S. Mead, M. Biggerstaff, L. Brammer, H. Davidson D. Jernigan, M.A. Jhung, L.A. Kamimoto, T.L Merlin, M. Nowell, S.C. Redd, C. Reed, A. Schuchat, and M.I. Meltzer,
Estimating the burden of 2009 pandemic influenza A (H1N1) in the United States (April 2009-April 2010),
\textit{Clin. Infect. Dis.}, 
\textbf{Jan 1;52 Suppl 1:S75-82} (2011).

\bibitem{webb1981reaction} G.F. Webb, A reaction-diffusion model for a deterministic diffusive epidemic, {\it Journal of Mathematical Analysis and Applications}, \textbf{84(1)} (1981), 150-161.

\end{thebibliography}
\end{document}